\documentclass[british]{amsart}
\usepackage[T1]{fontenc}
\usepackage[latin9]{inputenc}
\usepackage{color}
\usepackage{verbatim}
\usepackage{mathrsfs}
\usepackage{amsbsy}
\usepackage{amsthm}
\usepackage{amssymb}

\makeatletter
\numberwithin{figure}{section}
  \theoremstyle{plain}
  \newtheorem{lem}{\protect\lemmaname}
\theoremstyle{plain}
\newtheorem{thm}{\protect\theoremname}
 \theoremstyle{definition}
  \newtheorem{example}{\protect\examplename}

\usepackage{upgreek}
\usepackage{amssymb}
\usepackage{mathtools}
\let\originalleft\left
\let\originalright\right
\renewcommand{\left}{\mathopen{}\mathclose\bgroup\originalleft}
\renewcommand{\right}{\aftergroup\egroup\originalright}

\makeatother

\usepackage{babel}
  \providecommand{\examplename}{Example}
  \providecommand{\lemmaname}{Lemma}
\providecommand{\theoremname}{Theorem}

\begin{document}

\title[Algebraic Foundation of Dimensional Analysis]{An Algebraic Foundation of\\
Amended Dimensional Analysis}

\author{Dan Jonsson}

\address{Dan Jonsson, University of Gothenburg, Gothenburg, Sweden }

\email{dan.jonsson@gu.se}
\begin{abstract}
We present an innovative approach to dimensional analysis, based on
a general representation theorem for complete quantity functions admitting
a covariant scalar representation; this theorem is in turn grounded
in a purely algebraic theory of quantity spaces. Examples of dimensional
analysis based on this approach are given, showing that it allows
results obtained by traditional dimensional analysis to be strengthened.
For example, the orbital period of a two-body system can be derived
without use of equations of motion.
\end{abstract}

\maketitle

\section{Introduction\label{sec:I}}

The central theorem in dimensional analysis is the so-called $\pi$
theorem, with a long history featuring contributions by Fourier \cite{key-4},
Vaschy \cite{key-8}, Federman \cite{key-3}, Buckingham \cite{key-2}
and others. The $\pi$ theorem shows how to transform a ''physically
meaningful'' equation 
\begin{equation}
t=\phi\left(t_{1},\ldots,t_{n}\right),\label{eq:pi1}
\end{equation}
describing a relationship among quantities, into a more informative
equation. This is done by representing $\phi$ as a product of the
form $\prod_{j=1}^{r}\nolimits\!x{}_{j}^{W_{j}}\psi$ so that, specifically,
\begin{equation}
t=\prod_{j=1}^{r}\nolimits\!x{}_{j}^{W_{j}}\,\psi\left(y_{1}^{*},\ldots,y_{n-r}^{*}\right),\label{pi2}
\end{equation}
where $y_{k}^{\ast}=y_{k}\prod_{j=1}^{r}\!x{}_{j}^{-W_{kj}}$ is a
''dimensionless product'' usually denoted $\Pi_{k}$ or $\pi_{k}$,
and $\left\{ \left(x_{1},\ldots,x_{r}\right),\left(y_{1},\ldots,y_{n-r}\right)\right\} $
is a particular partition of $\left(t_{1},\ldots,t_{n}\right)$.

In traditional dimensional analysis, $\phi$ is a real function, and
$t,t_{1},\ldots,t_{n}$ are measures of physical quantities. Various
assumptions pertaining to $\phi$, $\psi$, $t$, $t_{i}$, $W_{j}$,
and $W_{kj}$ have been made in connection with the development of
the theory of dimensional analysis:
\begin{itemize}
\item Early proofs of the $\pi$ theorem relied on assumptions that $\phi$
is well-behaved in terms of continuity or smoothness \cite[pp. 67--71]{key-9}.
Langhaar \cite{key-11} and Brand \cite{key-12} later showed that
a generalised homogeneity assumption suffices.
\item It is usually assumed that $t,t_{i}>0$. This assumption is necessary
in some proofs of the $\pi$ theorem, and avoids the anomaly that
terms such as $\left(-1\right)^{\frac{1}{2}}$ are not real numbers,
but it restricts the scope of dimensional analysis.
\item The exponents $W_{j}$ and $W_{kj}$ are usually assumed to be rational
or real numbers \cite[p. 293]{key-13}, but Quade \cite{key-15} and
more recently Raposo \cite{key-19} use  integer exponents (see also
\cite{key-1}).
\item It is usually implicitly assumed that for any $\phi$ there is just
one $\psi$ such that (\ref{pi2}) holds, or at least that it suffices
to consider one $\psi$ or, at the very least, deal with one $\psi$
at a time. This way of thinking is challenged in \cite{key-5} (see
also \cite{key-1,key-6,key-27}).
\end{itemize}
Following the development of quantity calculus \cite{key-2-1}, some
versions of the $\pi$ theorem where $t,t_{1},\ldots,t_{n}$ are the
quantities measured rather than the numerical measures obtained have
also been proposed \cite{key-14,key-15,key-16,key-24,key-19}. In
that context, too, assumptions of the type highlighted above need
to be addressed.

This article connects to and synthesises some recent developments
in the theory of dimensional analysis \cite{key-1,key-5,key-7,key-26,key-22}.
To summarise, the new approach to dimensional analysis presented here
is characterised by the following features:
\begin{itemize}
\item Rather than a real function $\phi$ we use a quantity function $\Phi$
on a quantity space $Q$, as defined in Section \ref{sec:2}, so that
$t,t_{1},\ldots,t_{n}$ are seen as quantities. The only assumption
needed about $\Phi$ is that it is a quantity function admitting a
covariant scalar representation, as described below.
\item $Q$ can be a quantity space over any field, not only $\mathbb{R}$,
so the (generalised) measures of the quantities $t,t_{1},\ldots,t_{n}$
need not be real numbers. The quantities $x_{1},\ldots,x_{r}$ are
assumed to be non-zero, but there are no further constraints on $t,t_{1},\ldots,t_{n}$
(or their measures).
\item In equations (\ref{eq:pi1}) and (\ref{pi2}), $t$ is replaced by
$t^{W}$, where $W$ is a positive integer. In equation (\ref{pi2}),
$W_{j}$ are integers and we have $y_{k}^{*}=y_{k}^{W_{k}}\prod_{j=1}^{r}\!x_{j}^{-W_{kj}}$
, where $W_{k}$ are positive integers and $W_{kj}$ integers.
\item There may be more than one partition $\left\{ \left(x_{1},\ldots,x_{r}\right),\left(y_{1},\ldots,y_{n-r}\right)\right\} $
of $\left(t_{1},\ldots,t_{n}\right)$ such that $t^{W}=\prod_{j=1}^{r}\nolimits x{}_{j}^{W_{j}}\,\Psi\left(y_{1}^{*},\ldots,y{}_{n-r}^{*}\right)$
holds for particular $W$, $W_{j}$, $W_{k}$, $W_{kj}$ and $\Psi$,
so we actually have a non-empty set of equations
\[
t^{W_{\!\left(\ell\right)}}=\prod_{j=1}^{r}\nolimits x{}_{\left(\ell\right)j}^{W_{\!\left(\ell\right)j}}\,\Psi_{\ell}\left(y_{\left(\ell\right)1}^{*},\ldots,y{}_{\left(\ell\right)n-r}^{*}\right)\qquad\left(\ell=1,\ldots,S\right).
\]
\end{itemize}
After some preliminaries in Section 2, the main representation theorem
is stated and proved in Section 3 and applied to problems of dimensional
analysis in Sections 4 and 5. Examples of dimensional analysis are
found in Section 6, and physical assumptions underlying dimensional
analysis are touched on in Section 7.

\section{Preliminaries\label{sec:2}}

For the sake of completeness, let us first briefly review some elements
of the theory of quantity spaces \cite{key-5,key-7}. A \emph{scalable
monoid} over a ring $R$ is a monoid $Q$ equipped with an $R$-action
$\cdot$ on $Q$, 
\[
R\times Q\rightarrow Q,\qquad\left(\alpha,x\right)\mapsto\alpha\cdot x,
\]
such that for any $\alpha,\beta\in R$ and $x,y\in Q$ we have $1\cdot x=x$,
$\alpha\cdot\left(\beta\cdot x\right)=\alpha\beta\cdot x$, and $\alpha\cdot xy=\left(\alpha\cdot x\right)y=x\left(\alpha\cdot y\right)$;
as a consequence, $\left(\alpha\cdot x\right)\left(\beta\cdot y\right)=\alpha\beta\cdot xy$.
We denote the identity element of $Q$ by $1_{\!Q}$, and set $x^{0}=1_{\!Q}$
for any $x\in Q$. An element $x\in Q$ may have an inverse $x^{-1}\in Q$
such that $xx^{-1}=x^{-1}x=1_{\!Q}$.

A (strong) finite \emph{basis} for a scalable monoid $Q$ is a set
$\mathcal{E}=\left\{ e_{1},\ldots,e_{m}\right\} $ of invertible elements
of $Q$ such that every $x\in Q$ has a unique expansion 
\[
x=\mu_{\mathcal{E}}\left(x\right)\cdot\prod_{j=1}^{m}\nolimits e_{j}^{\mathfrak{W}_{\left(x\right)j}},
\]
where $\mu_{\mathcal{E}}\left(x\right)\in R$ and $\mathfrak{W}_{\left(x\right)j}$
are integers. A (finitely generated) \emph{quantity space} is a commutative
scalable monoid $Q$ over a field $K$, such that there exists a finite
basis for $Q$. The elements of a quantity space are called \emph{quantities}.
We may think of $\mu_{\mathscr{\mathcal{E}}}\left(x\right)$ as the
measure of $x$ relative to the derived unit $\prod_{j=1}^{m}\nolimits e_{j}^{\mathfrak{W}_{\left(x\right)j}}$
in a coherent system of units, and indirectly relative to the base
units in $\mathcal{E}$.

The relation $\sim$ on $Q$ defined by $x\sim y$ if and only if
$\alpha\cdot x=\beta\cdot y$ for some $\alpha,\beta\in K$ is a congruence
on $Q$. The corresponding equivalence classes are called \emph{dimensions};
$\left[x\right]$ is the dimension that contains $x$. We have $\left[\lambda\cdot x\right]=\left[x\right]$
for any $\lambda\in K,x\in Q$. The set of all dimensions in $Q$,
denoted $Q/{\sim}$, is a finitely generated free abelian group with
multiplication defined by $\left[x\right]\left[y\right]=\left[xy\right]$
and identity element $\left[1_{Q}\right]$. 

The quantities in a dimension $\mathsf{C}\in Q/{\sim}$ form a one-dimensional
vector space over $K$ with addition inherited from $K$ and a unique
zero vector or \emph{zero} \emph{quantity} $0_{\mathsf{C}}\neq1_{Q}$.
While $0_{\mathsf{C}}x=0_{\mathsf{C}\left[x\right]}$ for every $x\in Q$,
the product of non-zero quantities is a non-zero quantity. A quantity
is invertible if and only if it is non-zero, and any $u\in\mathsf{C}$
such that $u\neq0_{\mathsf{C}}$ is a \emph{unit quantity} for $\mathsf{C}$,
meaning that for every $x\in\mathsf{C}$ there is a unique $\mu\in K$
for $u$ such that $x=\mu\cdot u$, where $\mu=0$ if and only if
$x=0_{\mathsf{C}}$. %

We can now state some important definitions relating to dimensional
analysis. A (dimensional)\emph{ quantity function} on a quantity space
$Q$ over $K$ is a function 
\begin{equation}
\Phi:\mathsf{C}_{1}\times\cdots\times\mathsf{C}_{n}\rightarrow\mathsf{C}_{0},\label{eq:qfunc}
\end{equation}
where $\mathsf{C}_{i}\in Q/{\sim}$ for $i=0,\ldots,n$. Equation
(\ref{eq:pi1}) now takes the form 
\begin{equation}
q=\Phi\left(q_{1},\ldots,q_{n}\right).\label{eq:quantfunc}
\end{equation}

A \emph{dimensionally complete} quantity function is a quantity function
of the form (\ref{eq:qfunc}) such that there is a subtuple $\left(\mathsf{E}_{1},\ldots,\mathsf{E}_{r}\right)$
of $\left(\mathsf{C}_{1},\ldots,\mathsf{C}_{n}\right)$, a \emph{local
dimensional basis}, such that $r\geq0$ and every $\mathsf{C}_{i}$
has a unique expansion, with $\mathfrak{W}_{ij}$ integers, 
\[
\;\mathsf{C}_{i}=\prod_{j=1}^{r}\nolimits\!\mathsf{E}_{j}^{\mathfrak{W}_{ij}}\;\left(r\geq1\right),\qquad\mathsf{C}_{i}=\prod_{j=1}^{0}\nolimits\!\mathsf{E}_{j}^{\mathfrak{W}_{ij}}=\left[1_{Q}\right].
\]

A \emph{complete} quantity function is a dimensionally complete quantity
function with a local dimensional basis $\left(\mathsf{E}_{1},\ldots,\mathsf{E}_{r}\right)$
and a corresponding \emph{local basis}, a tuple of non-zero quantities
$E=\left(e_{1},\ldots,e_{r}\right)\in\bigtimes_{j=1}^{r}\!\mathsf{E}_{j}$,
where $E=\left(1_{Q}\right)$ for $r=0$, such that every $q_{i}\in\mathsf{C}_{i}$
has a unique expansion, with $\mathcal{W}{}_{ij}$ integers and $\mu_{E}\left(q_{i}\right)\in K$,
\[
q_{i}=\mu_{E}\left(q_{i}\right)\cdot\prod_{j=1}^{r}\nolimits\!e{}_{j}^{\mathcal{W}{}_{ij}}\;\left(r\geq1\right),\qquad q_{i}=\mu_{E}\left(q_{i}\right)\cdot\prod_{j=1}^{0}\nolimits\!e{}_{j}^{\mathcal{W}{}_{ij}}\!=\mu_{\left(1_{Q}\right)\!}\left(q_{i}\right)\cdot1_{Q}.
\]

A \emph{covariant scalar representation} of a complete quantity function
$\Phi$ of the form (\ref{eq:qfunc}) is a function $\upphi:K^{n}\rightarrow K$
such that 
\begin{equation}
\mu_{E}\left(q\right)=\upphi\left(\mu_{E}\left(q_{1}\right),\ldots,\mu_{E}\left(q_{n}\right)\right)\label{eq:covscalfunc}
\end{equation}
for every local basis $E$ associated with $\Phi$ and any $q_{1},\ldots,q_{n}$.

Alternatively, one could define a covariant scalar representation
in terms of a (global) basis $\mathcal{E}$ for $Q$ rather than a
local basis $E$ (see \cite{key-5}). The common idea is that $\Phi$
is not affected by a choice of (base) units, and for $\upphi$ to
similarly remain unchanged when units change the left-hand and right-hand
sides of (\ref{eq:covscalfunc}) and its $\mu_{\mathcal{E}}$-analogue
must change in tandem so that the two sides remain equal.

\textcolor{black}{Note that a complete quantity function need not
have a covariant scalar representation. For example, if $\mathsf{C}\neq\left[1_{Q}\right]$
and $0_{\mathsf{C}}$ $\neq e\in\mathsf{C}$ then any quantity function
$\Phi:\mathsf{C}\rightarrow\left[1_{Q}\right]$ is complete with local
dimensional basis $\left(\mathsf{C}\right)$ and local basis $\left(e\right)$
since $q_{1}=\mu_{\left(e\right)}\left(q_{1}\right)\cdot e^{1}$ and
$q=\mu_{\left(e\right)}\left(q\right)\cdot e^{0}$ are unique expansions
(recall that $1_{Q}\neq0_{\mathsf{\left[1_{Q}\right]}})$, but if
we set $\Phi_{u}\left(\lambda\cdot u\right)=\lambda\cdot1_{Q}$ for
some $0_{\mathsf{C}}\neq u\in\mathsf{C}$ then $\mu_{\left(u\right)}\left(u\right)=\mu_{\left(2\cdot u\right)}\left(2\cdot u\right)=1$,
whereas $\mu_{\left(u\right)}\left(\Phi_{u}\left(u\right)\right)=1$
but $\mu_{\left(2\cdot u\right)}\left(\Phi_{u}\left(2\cdot u\right)\right)=2$
since $u^{0}=\left(2\cdot u\right)\vphantom{}^{0}=1_{Q}$. Hence,
there is no function $\upphi_{u}:K\rightarrow K$ such that $\mu_{E}\left(\Phi_{u}\left(q_{1}\right)\right)=\upphi_{u}\left(\mu_{E}\left(q_{1}\right)\right)$
for every $E$ and $q_{1}\in\mathsf{C}$. In fact, if $\mathsf{C}\neq\left[1_{Q}\right]$
and $\Phi:\mathsf{C}\rightarrow\left[1_{Q}\right]$ admits a covariant
scalar representation then $\Phi$ is a constant function; see Example
\ref{x1} in Section \ref{sec:6}.}

\section{Main theorem\label{sec:3}}

\textcolor{black}{In view of how these concepts were defined above,
a complete quantity function is necessarily dimensionally complete,
and a quantity function with a covariant scalar representation is
necessarily complete. On the other hand, by Theorem \ref{pimain},
corresponding to the $\pi$ theorem, every dimensionally complete
quantity function is complete, and if it has a covariant scalar representation
then it has a representation of the form shown in (\ref{pi2}). }

\textcolor{black}{Let us first verify some basic facts needed in the
proof of Theorem \ref{pimain}.}
\begin{lem}
\label{lem1}Let $Q$ be a quantity space over $K$ and let $E=\left(e_{1},\ldots,e_{r}\right)$
be a local basis such that $x=\mu_{E}\left(x\right)\cdot\prod_{j=1}^{r}\nolimits e{}_{j}^{\mathfrak{W}_{\left(x\right)j}}$
and $y=\mu_{E}\left(y\right)\cdot\prod_{j=1}^{r}\nolimits e{}_{j}^{\mathfrak{W}_{\left(y\right)j}}$
are the unique expansions of $x,y\in Q$ relative to $E$. 

Then (\emph{a}) $\mu_{E}\left(xy\right)=\mu_{E}\left(x\right)\mu_{E}\left(y\right)$,
(\emph{b}) if $x$ is invertible then $\mu_{E}\left(x^{-1}\right)=\frac{1}{\mu_{E}\left(x\right)}$,
(\emph{c}) if $x\in\left[1_{Q}\right]$ then $\mu_{E}\left(x\right)$
does not depend on $E$, and (\emph{d}) if $\left[x\right]=\left[y\right]$
and $y$ is invertible then $xy^{-1}\in\left[1_{Q}\right]$ and $\mu_{E}\left(xy^{-1}\right)=\frac{\mu_{E}\left(x\right)}{\mu_{E}\left(y\right)}$
does not depend on $E$.
\end{lem}
\begin{proof}
We will be using the fact that expansions of the form $q=\mu_{E}\left(q\right)\cdot\prod_{j=1}^{r}\nolimits e{}_{j}^{\mathfrak{W}_{\left(q\right)j}}$
are unique, given $E$. Recall that $\left(\mu\!\cdot\!u\right)\left(\mu'\!\cdot\!u'\right)\!=\!\mu\mu'\cdot uu'$,
so 
\[
xy\!=\!\mu_{E}\left(x\right)\mu_{E}\left(y\right)\cdot\prod_{j=1}^{r}\nolimits e{}_{j}^{\left(\mathfrak{W}_{\left(x\right)j}+\mathfrak{W}_{\left(y\right)j}\right)},
\]
proving (\emph{a}). Clearly, $1_{Q}=1\cdot\prod_{j=1}^{r}\nolimits e{}_{j}^{0}$
for any $E$, so 
\[
\mu_{E}\left(x\right)\mu_{E}\left(x^{-1}\right)=\mu_{E}\left(xx^{-1}\right)=\mu_{E}\left(1_{Q}\right)=1,
\]
proving (\emph{b}). Also, $1_{Q}=\prod_{j=1}^{r}\nolimits e{}_{j}^{0}$
is a unit quantity for $\left[1_{Q}\right]$, so if $x\in\left[1_{Q}\right]$
then $x=\mu\cdot1_{Q}=\mu_{E}\left(x\right)\cdot\prod_{j=1}^{r}\nolimits e{}_{j}^{0}$,
and $\prod_{j=1}^{r}\nolimits e{}_{j}^{0}$ does not depend on $E$,
so $\mu_{E}\left(x\right)$ does not depend on $E$, verifying (\emph{c}).
Finally, (\emph{d}) follows from the fact that $\left[xy^{-1}\right]=\left[x\right]\left[y^{-1}\right]=\left[x\right]\left[y\right]^{-1}=\left[1_{Q}\right]$
together with (\emph{a}), (\emph{b}), and (\emph{c}).
\end{proof}
\begin{thm}
\label{pimain}Let $Q$ be a quantity space over $K$, let $\Phi$
be a quantity function on $Q$, 
\begin{equation}
\Phi:\mathsf{C}_{1}\times\cdots\times\mathsf{C}_{n}\rightarrow\mathsf{C}_{0}\quad\left(n\geq0\right),\qquad\left(x_{1},\ldots,x_{r},y_{1},\ldots,y_{n-r}\right)\mapsto y_{0},\label{eq:regular-1}
\end{equation}
and let there be unique tuples of integers $\left(\mathfrak{W}_{i1},\ldots,\mathfrak{W}_{ir}\right)$
such that 
\begin{equation}
\mathsf{C}_{i}=\prod_{j=1}^{r}\nolimits\!\mathsf{C}_{j}^{\mathfrak{W}_{ij}}\qquad\left(i=0,\ldots,n\right).\label{eq:cond1}
\end{equation}
Then $\Phi$ is complete, and if $\Phi$ has a covariant scalar representation
then there exists a unique quantity function of $n-r$ arguments
\[
\Psi:\left[1_{Q}\right]\times\cdots\times\left[1_{Q}\right]\rightarrow\left[1_{Q}\right]
\]
such that if $x_{j}\neq0_{\mathsf{C}_{j}}$ for $j=1,\ldots,r$ then
\begin{gather}
\pi_{0}=\Psi\left(\pi_{1},\ldots,\pi_{n-r}\right),\label{eq:piteorem}
\end{gather}
or equivalently
\begin{equation}
y_{0}=\Phi\left(x_{1},\ldots,x_{r},y_{1},\ldots,y_{n-r}\right)=\prod_{j=1}^{r}\nolimits\!x_{j}^{W_{0j}}\,\Psi\left(\pi_{1},\ldots,\pi_{n-r}\right),\label{eq:pitheorem2}
\end{equation}
where $W_{0j}=\mathfrak{W}_{0j}$, $W_{kj}=\mathfrak{W}_{\left(k+r\right)j}$
for $k=1,\ldots,n-r$, and
\[
\pi_{k}=y_{k}\left(\prod_{j=1}^{r}\nolimits x_{j}^{W_{kj}}\right)^{-1}\qquad\left(k=0,\ldots,n-r\right).
\]
\end{thm}
\begin{proof}
Set $E=\left(e_{1},\ldots,e_{r}\right)$, where $0_{\mathsf{C}_{j}}\neq e_{j}\in\mathsf{C}_{j}$
for $j=1,\ldots r$. Then 
\[
\mathsf{C}_{i}=\prod_{j=1}^{r}\nolimits\mathsf{C}_{j}^{\mathfrak{W}_{ij}}=\prod_{j=1}^{r}\nolimits\left[e_{j}\right]^{\mathfrak{W}_{ij}}=\left[\prod\nolimits _{j=1}^{r}e{}_{j}^{\mathfrak{W}_{ij}}\right]\qquad\left(i=0,\ldots,n\right),
\]
so for every $q_{i}\in\mathsf{C}_{i}$ there is a unique $\mu_{E}\left(q_{i}\right)\in K$
for $\prod_{j=1}^{r}\nolimits e{}_{j}^{\mathfrak{W}_{ij}}$ such that
\begin{equation}
q_{i}=\mu_{E}\left(q_{i}\right)\cdot\prod_{j=1}^{r}\nolimits e{}_{j}^{\mathfrak{W}_{ij}}\qquad\left(i=0,\ldots,n\right)\label{eq:exp}
\end{equation}
because $0_{\mathsf{C}_{i}}\neq\prod_{j=1}^{r}e{}_{j}^{\mathfrak{W}_{ij}}\in\mathsf{C}_{i}$
so that $\prod_{j=1}^{r}\nolimits e{}_{j}^{\mathfrak{W}_{ij}}$ is
a unit quantity for $\mathsf{C}_{i}$. Also, if $q_{i}$ has another
expansion in terms of the integers $\mathfrak{W}_{ij}'$ then 
\[
\mathsf{C}_{i}=\prod_{j=1}^{r}\nolimits\mathsf{C}{}_{j}^{\mathfrak{W}_{ij}}=\left[\prod_{j=1}^{r}\nolimits e{}_{j}^{\mathfrak{W}_{ij}}\right]=\left[q_{i}\right]=\left[\prod_{j=1}^{r}\nolimits e{}_{j}^{\mathfrak{W}_{ij}'}\right]=\prod_{j=1}^{r}\nolimits\mathsf{C}{}_{j}^{\mathfrak{W}_{ij}'},
\]
 so the integers $\mathfrak{W}_{ij}$ in (\ref{eq:exp}) are unique
by the uniqueness of the expansion of $\mathsf{C}{}_{i}$.

{\small{}Given $\mathfrak{W}_{ij}$ (and thus $W_{kj}$) by (}\ref{eq:cond1}{\small{})
and }$0_{\mathsf{C}_{j}}\neq x_{j}\in\mathsf{C}_{j}$ for $j=1,\ldots r${\small{},
set} 
\[
\varDelta_{k}=\prod_{j=1}^{r}\nolimits x_{j}^{W_{kj}}\qquad\left(k=0,\ldots,n-r\right).
\]
Then $\varDelta_{k}$ is non-zero and $\left[\varDelta_{k}\right]=\prod_{j=1}^{r}\nolimits\mathsf{C}{}_{j}^{W_{kj}}=\left[y_{k}\right]$
since $\left[y_{0}\right]=\mathsf{C}_{0}=\prod_{j=1}^{r}\nolimits\mathsf{C}{}_{j}^{\mathfrak{W}_{0j}}=\prod_{j=1}^{r}\nolimits\mathsf{C}{}_{j}^{W_{0j}}$
and $\left[y_{k}\right]=\mathsf{C}_{k+r}=\prod_{j=1}^{r}\nolimits\mathsf{C}{}_{j}^{\mathfrak{W}_{\left(k+r\right)j}}=\prod_{j=1}^{r}\nolimits\mathsf{C}{}_{j}^{W_{kj}}$
for $k=1,\ldots,n-r$. Thus, $\varDelta_{k}^{-1}$ exists, and denoting
$\mu_{E}\left(y_{k}\varDelta_{k}^{-1}\right)$ by $\nu_{E}\left(y_{k}\right)$
we have, by Lemma \ref{lem1}(\emph{d}),
\[
\nu_{E}\left(y_{k}\right)=\mu_{E}\left(y_{k}\varDelta_{k}^{-1}\right)=\frac{\mu_{E}\left(y_{k}\right)}{\mu_{E}\left(\varDelta_{k}\right)}\qquad\left(k=0,\ldots,n-r\right),
\]
where $y_{k}\varDelta_{k}^{-1}\in\left[1_{Q}\right]$ and $\nu_{E}\left(y_{k}\right)$
does not depend on $E$. Also, by Lemma \ref{lem1}(\emph{a}), 
\[
\mu_{E}\left(\varDelta_{k}\right)=\mu_{E}\left(\prod_{j=1}^{r}\nolimits x{}_{j}^{W_{kj}}\right)=\prod_{j=1}^{r}\nolimits\mu_{E}\left(x_{j}\right)^{W_{kj}}\qquad\left(k=0,\ldots,n-r\right).
\]

Let $\boldsymbol{q}$ denote the sequence of quantities $x_{1},\ldots,x_{r},y_{1},\ldots,y_{n-r}$,
and let $\tau_{E}\left(\boldsymbol{q}\right)$ be the sequence of
scalars $\mu_{E}\left(x_{1}\right),\ldots,\mu_{E}\left(x_{r}\right),\mu_{E}\left(y_{1}\right),\ldots,\mu_{E}\left(y_{n-r}\right)$.
By assumption, there is a function $\upphi:K^{n}\rightarrow K$ such
that $\mu_{E}\left(\Phi\left(\boldsymbol{q}\right)\right)=\upphi\left(\tau_{E}\left(\boldsymbol{q}\right)\right)$
for any $\boldsymbol{q}$ and $E$, so as $x_{j}\neq0_{\mathsf{C}_{j}},\mu_{E}\left(x_{j}\right)\neq0$
for $j=1,\ldots r$ there is a function $\phi:K^{n}\mapsto K$ such
that for any $\boldsymbol{q}$ and $E$ we have 
\begin{gather*}
\nu_{E}\left(y_{0}\right)=\frac{\mu_{E}\left(\Phi\left(\boldsymbol{q}\right)\right)}{\mu_{E}\left(\varDelta_{0}\right)}=\frac{\upphi\left(\tau_{E}\left(\boldsymbol{q}\right)\right)}{\prod_{j=1}^{r}\mu_{E}\left(x_{j}\right)^{W_{0j}}}=\phi\left(\tau_{E}\left(\boldsymbol{q}\right)\right).
\end{gather*}
Furthermore, there is, for given $W_{kj}$, a bijection between scalar
sequences
\begin{gather*}
\omega:\tau_{E}\left(\boldsymbol{q}\right)\:\longmapsto\:\mu_{E}\left(x_{1}\right),\ldots,\mu_{E}\left(x_{r}\right),\frac{\mu_{E}\left(y_{1}\right)}{\prod_{j=1}^{r}\mu_{E}\left(x_{j}\right)^{W_{1j}}},\ldots,\frac{\mu_{E}\left(y_{n-r}\right)}{\prod_{j=1}^{r}\mu_{E}\left(x_{j}\right)^{W_{\left(n-r\right)j}}}\:,
\end{gather*}
so there is a function $\chi=\phi\circ\omega^{-1}:K^{n}\rightarrow K$
such that
\begin{equation}
\nu_{E}\left(y_{0}\right)=\phi\left(\tau_{E}\left(\boldsymbol{q}\right)\right)=\chi\left(\mu_{E}\left(x_{1}\right),\ldots,\mu_{E}\left(x_{r}\right),\nu_{E}\left(y_{1}\right),\ldots,\nu_{E}\left(y_{n-r}\right)\right).\label{eq:nuChi}
\end{equation}

Now set $X=\left(x_{1},\ldots,x_{r}\right)$. By assumption, $0_{\mathsf{C}_{j}}\neq x_{j}\in\mathsf{C}_{j}$
for $j=1,\ldots,r$, so any $q_{i}\in\mathsf{C}_{i}$ has a unique
expansion of the form (\ref{eq:exp}) relative to both $E$ and $X$,
and $\nu_{E}\left(y_{k}\right)=\nu_{X}\left(y_{k}\right)$ for any
$E$ and $k=0,\ldots,n-r$. Equation (\ref{eq:nuChi}) holds for $E=X$
as it holds for any $E$. There is thus a function $\uppsi:K^{n-r}\rightarrow K$
such that
\begin{align*}
\nu_{E}\left(y_{0}\right) & =\nu_{X}\left(y_{0}\right)=\chi\left(1,\ldots,1,\nu_{X}\left(y_{1}\right),\ldots,\nu_{X}\left(y_{n-r}\right)\right)\\
 & =\uppsi\left(\nu_{X}\left(y_{1}\right),\ldots,\nu_{X}\left(y_{n-r}\right)\right)=\uppsi\left(\nu_{E}\left(y_{1}\right),\ldots,\nu_{E}\left(y_{n-r}\right)\right)
\end{align*}
for any $E$, since $x_{j}=1\cdot x_{j}$ so that $\mu_{X}\left(x_{j}\right)=1$
for $j=1,\ldots,r$. 

$1_{Q}$ is a unit quantity for $\left[1_{Q}\right]$, and $\nu_{E}\left(y_{k}\right)$
does not depend on $E$, so we can define a quantity function of $n-r$
arguments 
\[
\Psi:\left[1_{Q}\right]\times\cdots\times\left[1_{Q}\right]\rightarrow\left[1_{Q}\right]
\]
by setting
\[
\Psi\left(\nu_{E}\left(y_{1}\right)\cdot1_{Q},\ldots,\nu_{E}\left(y_{n-r}\right)\cdot1_{Q}\right)=\uppsi\left(\nu_{E}\left(y_{1}\right),\ldots,\nu_{E}\left(y_{n-r}\right)\right)\cdot1_{Q},
\]
so that
\begin{equation}
\nu_{E}\left(y_{0}\right)\cdot1_{Q}=\Psi\left(\nu_{E}\left(y_{1}\right)\cdot1_{Q},\ldots,\nu_{E}\left(y_{n-r}\right)\cdot1_{Q}\right).\label{eq:prepi}
\end{equation}

\smallskip{}

Also, $q=\mu_{X}\left(q\right)\cdot\prod_{j=1}^{r}x_{j}^{0}$ is the
unique expansion of $q\in\left[1_{Q}\right]$ relative to $X$, so
\[
\nu_{E}\left(y_{k}\right)\cdot1_{Q}=\nu_{X}\left(y_{k}\right)\cdot1_{Q}=\mu_{X}\left(y_{k}\varDelta_{k}^{-1}\right)\cdot\prod_{j=1}^{r}\nolimits\!x{}_{j}^{0}=y_{k}\varDelta_{k}^{-1}
\]
for any $E$ and $k=0,\ldots,n-r$. Thus, (\ref{eq:prepi}) translates
to
\[
y_{0}\varDelta_{0}^{-1}=\Psi\left(y_{1}\varDelta_{1}^{-1},\ldots,y_{n-r}\varDelta_{n-r}^{-1}\right)\qquad\mathrm{or}\qquad\pi{}_{0}=\Psi\left(\pi{}_{1},\ldots,\pi{}_{n-r}\right);
\]
recall that $y_{k}\varDelta_{k}^{-1}=y_{k}\left(\prod_{j=1}^{r}\nolimits x_{j}^{W_{kj}}\right)^{-1}=\pi_{k}$
for $k=0,\ldots,n-r$. 

We have thus proved the existence of a function $\Psi$ that represents
$\Phi$ by equation (\ref{eq:piteorem}) or (\ref{eq:pitheorem2}).
Finally, if $\varDelta_{0}\Psi\left(\pi{}_{1},\ldots,\pi{}_{n-r}\right)=y_{0}=\varDelta_{0}\Psi'\left(\pi{}_{1},\ldots,\pi{}_{n-r}\right)$
then $\Psi\left(\pi{}_{1},\ldots,\pi{}_{n-r}\right)=\Psi'\left(\pi{}_{1},\ldots,\pi{}_{n-r}\right)$
since $\varDelta_{0}$ is invertible, so $\Psi$ is unique.
\end{proof}
In this proof, $E$ is is not assumed to be a basis for $Q$, but
only a local basis. This would seem to facilitate the generalization
of Theorem \ref{pimain} to certain other commutative scalable monoids
than finitely generated quantity spaces.

There is a result of the same kind as Theorem \ref{pimain} for ''physically
meaningful'' scalar functions. Specifically, if $\upphi:K^{n}\rightarrow K$
is a covariant scalar representation of a complete quantity function
$\Phi$ then there is a scalar function $\uppsi:K^{n-r}\rightarrow K$
such that 
\[
\upphi\left(\zeta_{1},\ldots,\zeta_{n}\right)=\prod_{j=1}^{r}\xi_{j}^{W_{0j}}\uppsi\left(\pi{}_{1},\ldots,\pi{}_{n-r}\right),
\]
where $\zeta{}_{i}=\mu_{E}\left(q_{i}\right)$, $\xi_{j}=\mu_{E}\left(x_{j}\right)$
and $\pi{}_{k}=\nu_{E}\left(y_{k}\right)$ for some local basis $E$.

The condition that $x_{1},\ldots,x_{r}$ are non-zero quantities is
natural and necessary; note that $y_{0},y_{1},\ldots,y_{n-r}$ are
not restricted. The commonly made assumption that quantities are positive
(or have positive measures) presupposes that $K$ is an ordered field
such as the real numbers, but the present representation theorem holds
for any field, for example, the complex numbers. Only real numbers
are used in the examples of dimensional analysis in Section \ref{sec:6},
however.

By definition, $y_{k}$ occurs only in $\pi_{k}$, so if the quantities
$y_{1},\ldots,y_{n-r}$ (but not necessarily their dimensions) are
independent (as they should be) then $\pi_{1},\ldots,\pi_{n-r}$ are
independent, too.

$Q$ has no zero divisors, so (\ref{eq:pitheorem2}) implies that
$\Phi\left(x_{1},\ldots,x_{r},y_{1},\ldots y_{n-r}\right)=0_{\mathsf{C}_{0}}$
if and only if $\Psi\left(\pi_{1},\ldots,\pi_{n-r}\right)=0_{\left[1_{Q}\right]}$,
given that $x_{j}\neq0_{\mathsf{C}_{j}}$ for $j=1,\ldots,r$. Thus,
a solution set of $n$-tuples of quantities is reduced to an equivalent
set of $n-r$-tuples of quantities. This equivalence corresponds to
the form of the $\pi$ theorem given by Vaschy \cite{key-8} and Buckingham
\cite{key-2}, whereas Federman \cite{key-3} proved an identity of
the form (\ref{eq:pitheorem2}). (Vaschy's and Buckingham's proofs
were sketchy, and Federman's proof covered only a special case, but
this was pioneering work.)

\section{Analysis of dimension tuples\label{sec:4}}

Given a quantity space $Q$, consider a \emph{dimension tuple} $\mathscr{D}=\left(\left[q_{1}\right],\ldots,\left[q_{n}\right],\left[q_{0}\right]\right)$,
where $q_{i}\in Q$\emph{ }and $n\geq0$. An (internally)\emph{ independent}
subtuple of $\mathscr{D}$ is a sub\-tuple $\left(\mathsf{D}_{1},\ldots,\mathsf{D}_{\rho}\right)$
of $\mathscr{D}$ such that if $\prod_{j=1}^{\rho}\mathsf{D}_{j}^{\kappa_{j}}=1_{Q}$
then $\kappa_{j}=0$ for $j=1,\ldots,\rho$. (Note that a subtuple
preserves the sequential order of entries in the original tuple.)\linebreak{}
 A\emph{ maximal independent} subtuple of $\mathscr{D}$ is an independent
subtuple of $\mathscr{D}$ which is not a proper subtuple of an independent
subtuple of $\mathscr{D}$. The empty subtuple $\emptyset$ of a dimension
tuple is vacuously independent. Thus, every $\mathscr{D}$ has a maximal
independent subtuple, possibly $\mathscr{D}$ itself, since every
non-repetitive chain of independent subtuples $\emptyset\subset\cdots\subset\mathscr{D}_{\rho}$
is finite, ending no later than when $\mathscr{D}_{\rho}=\mathscr{D}$. 

Recall that $Q/{\sim}$ is a free abelian group and that by the Steinitz
exchange lemma any two finite maximal independent subsets of a free
abelian group have the same cardinality, the rank of the group. An
independent subtuple of $\mathscr{D}$ cannot have duplicated entries,
so it can be regarded as a set in this context, and the exchange lemma
can be used also to show that any two maximal independent subtuples
of $\mathscr{D}$ have the same cardinality $r$, the \emph{rank of
the dimension tuple} $\mathscr{D}$. 

An\emph{ adequate} partition of $\mathscr{D}$ is a partition of $\mathscr{D}$
into three dimension tuples 
\begin{equation}
\mathscr{A}=\left(\left[x_{1}\right],\ldots,\left[x_{r}\right]\right),\quad\mathscr{B}=\left(\left[y_{1}\right],\ldots.\left[y_{n-r}\right]\right),\quad\mathscr{C}=\left(\left[y_{0}\right]\right)\label{eq:part1-1}
\end{equation}
such that $\mathscr{A}$ is a maximal independent subtuple of $\mathscr{D}$
and $n\geq r\geq0$. (The idea is that $\mathscr{A}$ is a local dimensional
basis and $y_{0}$ a dependent variable.) A dimension tuple $\mathscr{D}$
clearly admits an adequate partition unless the maximal independent
subtuple of $\mathscr{D}$ is $\mathscr{D}$ itself. There may in
fact be more than one adequate partition of $\mathscr{D}$, since
there may be more than one maximal independent proper subtuple $\mathscr{A}$
of $\mathscr{D}$, and for each $\mathscr{A}$ the given $\mathscr{D}$
may contain more than one entry not in $\mathscr{A}$. When $n>r$
so that $\mathscr{B}$ is non-empty, we can either choose $\left[y_{0}\right]$
first and then choose $\mathscr{A}$ as a maximal independent tuple
of other dimensions in $\mathscr{D}$ than $\left[y_{0}\right]$,
or we can choose $\mathscr{A}$ first and then choose $\left[y_{0}\right]$
among the remaining dimensions in $\mathscr{D}$.

If $\left(\left[x_{1}\right],\ldots,\left[x_{r}\right]\right)$ is
independent but $\left(\left[y_{k}\right],\left[x_{1}\right],\ldots,\left[x_{r}\right]\right)$
is not independent then there is a tuple of integers $\boldsymbol{w}_{k}=\left(w_{k},w_{k1},\ldots,w_{kr}\right)$
such that $\left[y_{k}\right]^{-w_{k}}\prod_{j=1}^{r}\nolimits\left[x_{j}\right]^{w_{kj}}=1_{Q}$
and $w_{k}\neq0$. Hence, if $\left\{ \mathscr{A},\mathscr{B},\mathscr{C}\right\} $
is an adequate partition of $\mathscr{D}$ then 
\begin{equation}
\left[y_{k}\right]^{w_{k}}=\prod_{j=1}^{r}\nolimits\left[x_{j}\right]^{w_{kj}},\label{eq:exp1-2}
\end{equation}
for $k=0,\ldots,n-r$. Without loss of generality, we may assume that
$w_{k}>0$.

Suppose that we also have 
\[
\left[y_{k}\right]^{w_{k}'}=\prod_{j=1}^{r}\nolimits\left[x_{j}\right]^{w_{kj}'}
\]
for a tuple of integers $\boldsymbol{w}_{k}'=\left(w_{k}',w_{k1}',\ldots,w_{kr}'\right)$
where $w_{k}'>0$. Furthermore, let $\boldsymbol{W}_{\!k}=\left(W_{k},W_{k1},\ldots,W_{kr}\right)$
be the tuple obtained by dividing $\boldsymbol{w}_{k}$ through by
$\gcd\boldsymbol{w}_{k}$, and $\boldsymbol{W}_{\!k}'=\left(W_{k}',W_{k1}',\ldots,W_{kr}'\right)$
be the tuple obtained by dividing $\boldsymbol{w}_{k}'$ through by
$\gcd\boldsymbol{w}_{k}'$. Then, $W_{k},W_{k}'>0$ and $\gcd\boldsymbol{W}_{\!k}=\gcd\boldsymbol{W}_{\!k}'=1$.
If $W_{k}=W_{k}'$ then
\[
\left[x_{1}\right]^{W_{k1}-W_{k1}'}\cdots\left[x_{r}\right]^{W_{kr}-W_{kr}'}=\left[y_{k}\right]^{W_{k}-W_{k}'}=\left[1_{Q}\right],
\]
so $\boldsymbol{W}_{\!k}=\boldsymbol{W}_{\!k}'$ since $\left(\left[x_{1}\right],\ldots,\left[x_{r}\right]\right)$
is independent. In any case, (\ref{eq:exp1-2}) implies that $\left[y_{k}\right]^{W_{k}W_{k}'}=\prod_{j=1}^{r}\nolimits\left[x_{j}\right]^{W_{kj}W_{k}'}$.
Dividing $\left(W_{k}W_{k}',W_{k1}W_{k}',\ldots,W_{kr}W_{k}'\right)$
through by $W_{k}$ gives $\boldsymbol{W}_{\!k}'$ since $\left(W_{k}W_{k}'\right)/W_{k}=W_{k}'$,
but as $\boldsymbol{W}_{\!k}'$ is an integer tuple this implies 
\[
W_{k}\leq\gcd\left(W_{k}W_{k}',W_{k1}W_{k}',\ldots,W_{kr}W_{k}'\right)=W_{k}'.
\]
By symmetry $W_{k}'\leq W_{k}$, so $W_{k}=W_{k}'$, so again $\boldsymbol{W}_{\!k}=\boldsymbol{W}_{\!k}'$.
There is thus, for every $k=0,\ldots,n-r$, a unique tuple of integers
$\left(W_{k},W_{k1},\ldots,W_{kr}\right)$ such that 
\begin{equation}
\left[y_{k}\right]^{W_{k}}=\prod_{j=1}^{r}\nolimits\left[x_{j}\right]^{W_{kj}},\label{eq:exp1-1-1-1}
\end{equation}
$W_{k}>0$ and $\gcd\left(W_{k},W_{k1},\ldots,W_{kr}\right)=1$. 

It follows from this result that $\left(w_{k},w_{k1},\ldots,w_{kr}\right)$,
with $w_{k}>0$, satisfies (\ref{eq:exp1-2}) if and only if $\left(w_{k},w_{k1},\ldots,w_{kr}\right)\!=\!\left(\kappa W_{k},\kappa W_{k1},\ldots,\kappa W_{kr}\right)$
for some positive integer $\kappa$.

Also, for $j=1,\ldots,r$, $\left[x_{j}\right]$ has the trivial expansion
\begin{equation}
\left[x_{j}\right]=\left[x_{1}\right]^{0}\cdots\left[x_{j-1}\right]^{0}\left[x_{j}\right]^{1}\left[x_{j+1}\right]^{0}\cdots\left[x_{r}\right]^{0},\label{eq:trivexp}
\end{equation}
relative to $\mathscr{A}=\left(\left[x_{1}\right],\ldots,\left[x_{r}\right]\right)$,
and this expansion is unique since $\mathscr{A}$ is independent. 

Thus, given an adequate partition (\ref{eq:part1-1}) of $\mathscr{D}=\left(\left[q_{1}\right],\ldots,\left[q_{n}\right],\left[q_{0}\right]\right)$,
any function
\[
\Phi^{*}:\left[x_{1}\right]\times\cdots\times\left[x_{r}\right]\times\left[y_{1}^{W_{1}}\right]\times\cdots\times\left[y_{n-r}^{W_{n-r}}\right]\rightarrow\left[y_{0}^{W_{0}}\right]
\]
is dimensionally complete in view of the unique expansions (\ref{eq:exp1-1-1-1})
and (\ref{eq:trivexp}) and the fact that $\left[y_{k}^{W_{k}}\right]=\left[y_{k}\right]^{W_{k}}$.
It thus follows from Theorem \ref{pimain} that $\Phi^{*}$ is complete
and that, assuming that $\Phi^{*}$ admits a covariant scalar representation
$\upphi^{*}$, there exists a unique quantity function of $n-r$ arguments
\[
\Psi:\left[1_{Q}\right]\times\cdots\times\left[1_{Q}\right]\rightarrow\left[1_{Q}\right]
\]
such that if $x_{j}\neq0_{\left[x_{j}\right]}$ for $j=1,\ldots,r$
then 
\[
y_{0}^{W_{0}}=\Phi^{*}\left(x_{1},\ldots,x_{r},y_{1}^{W_{1}},\ldots,y_{n-r}^{W_{n-r}}\right)=\prod_{j=1}^{r}\nolimits\!x_{j}^{W_{0j}}\,\Psi\left(\pi_{1},\ldots,\pi_{n-r}\right),
\]
where $\pi_{k}=y_{k}^{W_{k}}\left(\prod_{j=1}^{r}\nolimits x_{j}^{W_{kj}}\right)^{\!-1}$for
$k=1,\ldots,n-r$. Given $\Phi^{*}$, there is also a unique function
\[
\Phi:\left[x_{1}\right]\times\cdots\times\left[x_{r}\right]\times\left[y_{1}\right]\times\cdots\times\left[y_{n-r}\right]\rightarrow\left[y_{0}^{W_{0}}\right]
\]
such that 
\[
\Phi\left(x_{1},\ldots,x_{r},y_{1},\ldots,y_{n-r}\right)=\Phi^{*}\left(x_{1},\ldots,x_{r},y_{1}^{W_{1}},\ldots,y_{n-r}^{W_{n-r}}\right).
\]
Thus,
\begin{equation}
y_{0}^{W_{0}}=\Phi\left(x_{1},\ldots,x_{r},y_{1},\ldots,y_{n-r}\right)=\prod_{j=1}^{r}\nolimits\!x_{j}^{W_{0j}}\,\Psi\left(\pi_{1},\ldots,\pi_{n-r}\right),\label{eq:piGen}
\end{equation}
where $\pi_{k}$ is defined as above; note that (\ref{eq:piGen})
generalises (\ref{eq:pitheorem2}) in Theorem \ref{pimain}.

We have thus shown that a non-empty set of adequate partitions of
$\mathscr{D}$ such that $\mathscr{C}=\left(\left[y_{0}\right]\right)$
for some fixed $\left[y_{o}\right]\in\mathscr{D}$ yields a set of
$S\geq1$ equations
\begin{equation}
\begin{cases}
y_{0}^{W_{\!\left(1\right)0}}=\prod_{j=1}^{r}\nolimits\!x{}_{\left(1\right)j}^{W_{\!\left(1\right)0j}}\,\Psi_{1}\left(\pi{}_{\left(1\right)1},\ldots,\pi{}_{\left(1\right)n-r}\right),\\
\qquad\cdots\\
y_{0}^{W_{\!\left(S\right)0}}=\prod_{j=1}^{r}\nolimits\!x{}_{\left(S\right)j}^{W_{\!\left(S\right)0j}}\,\Psi_{S}\left(\pi{}_{\left(S\right)1},\ldots,\pi{}_{\left(S\right)n-r}\right),
\end{cases}\label{eq:qeqsyst}
\end{equation}
where $r$ is the rank of $\mathscr{D}$. Each such equation corresponds
to a maximal independent subtuple $\mathscr{A}$ of $\mathscr{D}$
that does not contain $\left[y_{0}\right]$.

If (\ref{eq:qeqsyst}) is not already a system of simultaneous equations,
it can be made into one by setting $\varLambda_{\left(\ell\right)}=\mathrm{lcm}\left(W_{\!\left(1\right)0},\ldots,W{}_{\!\left(S\right)0}\right)/W_{\!\left(\ell\right)0}$
for $\ell=1,\ldots,S$ and replacing the $\ell$th equation in (\ref{eq:qeqsyst})
by 
\begin{equation}
y_{0}^{W_{0}}=\prod_{j=1}^{r}\nolimits\!x{}_{\left(\ell\right)j}^{\varLambda_{\left(\ell\right)}W{}_{\!\left(\ell\right)0j}}\,\Psi_{\ell}\left(\pi{}_{\left(\ell\right)1},\ldots,\pi{}_{\left(\ell\right)n-r}\right)^{\varLambda_{\left(\ell\right)}},\label{eq:normeqsyst}
\end{equation}
where $W_{0}=\mathrm{lcm}\left(W_{\!\left(1\right)0},\ldots,W{}_{\!\left(S\right)0}\right)=\varLambda_{\left(\ell\right)}W_{\!\left(\ell\right)0}$
for $\ell=1,\ldots,S$. 

For convenience, we may use notation similar to that in Section \ref{sec:I},
simplifying $y_{0}^{W_{0}}$ to $y^{W}$, $y_{0}^{W_{\left(\ell\right)0}}$
to $y^{W_{\!\left(\ell\right)}}$ and $W_{\!\left(\ell\right)0j}$
to $W_{\!\left(\ell\right)j}$.

\textcolor{black}{Given $\Phi$, $\Phi^{*}$ and $\upphi^{*}$, we
define the covariant scalar representation of $\Phi$ to be the function
$\upphi:K^{n}\rightarrow K$ given by 
\[
\upphi\left(\xi_{1},\ldots,\xi_{r},\eta_{1},\ldots,\eta_{n-r}\right)=\upphi^{*}\left(\xi_{1},\ldots,\xi_{r},\eta_{1}^{W_{1}},\ldots,\eta_{n-r}^{W_{n-r}}\right).
\]
As $\Phi$ need not be complete, although it is associated with a
complete quantity function $\Phi^{*}$, we have thus extended the
notion of a covariant scalar representation. By this extended definition,
$\Phi$ has a covariant scalar representation if and only if there
is a corresponding complete quantity function $\Phi^{*}$ with a covariant
scalar representation. }

\textcolor{black}{While the original definition of a covariant scalar
representation was appropriate in the context of Theorem \ref{pimain},
the generalised definition makes it possible to restate the core conclusion
from the argument presented in this section as an apparent generalisation
of Theorem \ref{pimain}, derived, however, by an application of this
theorem. Similar results, which may likewise be regarded as reformulations
of Theorem 1, are found in \cite[pp. 16--19]{key-6} and \cite[pp. 75--76]{key-19}.}
\begin{thm}
\textcolor{black}{\label{piAlt}Let $Q$ be a quantity space over
$K$, let $\Phi$ be a quantity function on $Q$, 
\[
\Phi:\mathsf{C}_{1}\times\cdots\times\mathsf{C}_{n}\rightarrow\mathsf{C}_{0}\quad\left(n\geq0\right),\qquad\left(q_{1},\ldots,q_{n}\right)\mapsto q_{0},
\]
let $\mathscr{D}=\left(\left[q_{1}\right],\ldots,\left[q_{n}\right],\left[y_{0}\right]\right)$
be a dimension tuple such that there exists a minimal positive integer
$W$ such that $y_{0}^{W}=q_{0}$, and let 
\[
\mathscr{A}=\left(\left[x_{1}\right],\ldots,\left[x_{r}\right]\right),\quad\mathscr{B}=\left(\left[y_{1}\right],\ldots.\left[y_{n-r}\right]\right),\quad\mathscr{C}=\left(\left[y_{0}\right]\right)
\]
be an adequate partition of $\mathscr{D}$. Then there are unique
tuples of integers
\[
\left(W_{k},W_{k1},\ldots,W_{kr}\right)\qquad\left(k=0,\ldots,n-r\right),
\]
where $W_{k}>0$ and $\gcd\left(W_{k},W_{k1},\ldots,W_{kr}\right)=1$,
such that if $\Phi$ has a covariant scalar representation then there
exists a unique quantity function of $n-r$ arguments
\[
\Psi:\left[1_{Q}\right]\times\cdots\times\left[1_{Q}\right]\rightarrow\left[1_{Q}\right]
\]
such that if $x_{j}\neq0_{\mathsf{C}_{j}}$ for $j=1,\ldots,r$ then
\[
y_{0}^{W_{0}}=\Phi\left(q_{1},\ldots,q_{n}\right)=\prod_{j=1}^{r}\nolimits\!x_{j}^{W_{0j}}\,\Psi\left(\pi_{1},\ldots,\pi_{n-r}\right),
\]
where $W_{0}=W$ and 
\[
\pi_{k}=y_{k}^{W_{k}}\left(\prod_{j=1}^{r}\nolimits x_{j}^{W_{kj}}\right)^{-1}\qquad\left(k=1,\ldots,n-r\right).
\]
}
\end{thm}
\textcolor{black}{Note that if $\left[q_{0}\right]\neq\left[1_{Q}\right]$
and $y_{0}^{w}=q_{0}$ for some $w>0$ then $W=w$ as $w$ is unique,
but if $\left[q_{0}\right]=\left[1_{Q}\right]$ and $y_{0}^{w}=q_{0}$
for some $w>0$ then $W=1$.}

\section{Use of dimensional matrices\label{sec:5}}

\textcolor{black}{We have shown that, given a dimension tuple $\mathscr{D}=\left(\left[q_{1}\right],\ldots,\left[q_{n}\right],\left[y_{0}\right]\right)$,
the information needed to derive the equations in (\ref{eq:qeqsyst})
is a list of all adequate partitions $\left\{ \mathscr{A},\mathscr{B},\mathscr{C}=\left(\left[y_{0}\right]\right)\right\} $
of $\mathscr{D}$ and, for each $\mathscr{A}=\left(\left[x_{1}\right],\ldots,\left[x_{r}\right]\right)$
and each $\left[y_{k}\right]$ in $\mathscr{D}$ but not in $\mathscr{A}$,
a tuple of integers $\left(w_{k},w_{k1},\ldots,w_{kr}\right)$ such
that $w_{k}\neq0$ and $\left[y_{k}\right]^{w_{k}}=\prod_{j=1}^{r}\nolimits\left[x_{j}\right]^{w_{kj}}$;
as $\left[y_{k}\right]^{-w_{k}}=\prod_{j=1}^{r}\nolimits\left[x_{j}\right]^{-w_{kj}}$
we may assume that $w_{k}>0$.}

In dimensional analysis, narrowly conceived, this information comes
from a \linebreak{}
\emph{dimensional matrix} showing how the dimensions in $\mathscr{D}$
are expressed as products of dimensions of certain base units. Specifically,
let $\mathscr{E}=\left\{ \mathsf{E}_{1},\ldots,\mathsf{E}_{m}\right\} $
be a basis for $Q/{\sim}$. Then any entry $\left[q_{i}\right]=\mathsf{D}_{i}$
in $\mathscr{D}$ can be expressed uniquely as
\begin{equation}
\mathsf{D}_{i}=\prod_{\ell=1}^{m}\nolimits\mathsf{E}_{\ell}^{D_{i\ell}}\qquad\mathrm{or}\qquad\mathsf{D}_{i}=\prod_{\ell=1}^{m}\nolimits\mathsf{E}_{\ell}^{D_{\ell i}}\qquad\left(i=0,\ldots,n\right).\label{eq:dimcol}
\end{equation}
There are corresponding dimensional matrices 
\[
\left[D_{i\ell}\right]=\left[\begin{array}{ccc}
\cdots &  & \cdots\\
D_{i1} & \cdots & D_{im}\\
\cdots &  & \cdots
\end{array}\right],\qquad\left[D_{\ell i}\right]=\left[\begin{array}{ccc}
\cdots & D_{1i} & \cdots\\
 & \cdots\\
\cdots & D_{mi} & \cdots
\end{array}\right],
\]
which collect the exponents that describe the expansions of $\mathsf{D}_{0},\ldots,\mathsf{D}_{n}$
relative to $\mathscr{E}$; we will use $\left[D_{\ell i}\right]$
here. Corresponding to an adequate partition of $\mathscr{D}$ into
$\mathscr{A}$, $\mathscr{B}$ and $\mathscr{C}$, $\left[D_{\ell i}\right]$
can be partitioned into three matrices,
\[
\left[A_{\ell i}\right],\quad\left[B_{\ell i}\right],\quad\left[C_{\ell i}\right],
\]
constituting a \emph{dimensional model} derived from $\left[D_{\ell i}\right]$. 

Let $\mathbf{u}_{i}$ denote the column vector $\left(D_{1i},\ldots,D_{mi}\right)^{\!\mathrm{T}}$;
in view of (\ref{eq:dimcol}) there is a bijection $\mathsf{D}_{i}\mapsto\mathbf{u}_{i}$
between the dimensions in $\mathscr{D}$ and the column vectors in
$\left[D_{\ell i}\right]$. Also let $\left(\mathsf{D}_{i_{1}},\ldots,\mathsf{D}_{i_{\rho}}\right)$
be a subtuple of $\mathscr{D}$, corresponding to a tuple of column
vectors $\left(\mathbf{u}_{i_{1}},\ldots,\mathbf{u}_{i_{\rho}}\right)$.
To avoid notational clutter, we write $\mathsf{D}_{i_{j}}$ as $\overline{\mathsf{D}}_{j}$,
$\mathbf{u}_{i_{j}}$ as $\overline{\mathbf{u}}_{j}$, $D_{\ell i_{j}}$
as $\overline{D}_{\ell j}$ and $\kappa_{i_{j}}$ as $\overline{\kappa}_{j}$.
By (\ref{eq:dimcol}), we have
\[
\prod_{j=1}^{\rho}\nolimits\!\overline{\mathsf{D}}_{j}^{\,\overline{\kappa}_{j}}\!=\prod_{j=1}^{\rho}\nolimits\left(\prod_{\ell=1}^{m}\nolimits\!\mathsf{E}_{\ell}^{\overline{D}_{\ell j}}\right)^{\!\overline{\kappa}_{j}}\!=\prod_{\ell=1}^{m}\nolimits\prod_{j=1}^{\rho}\nolimits\!\mathsf{E}_{\ell}^{\overline{D}_{\ell j}\overline{\kappa}_{j}}\!=\prod_{\ell=1}^{m}\nolimits\!\mathsf{E}_{\ell}^{\sum_{j=1}^{\rho}\overline{D}_{\ell j}\overline{\kappa}_{j}}.
\]
Thus, $\prod_{j=1}^{\rho}\nolimits\!\overline{\mathsf{D}}_{j}^{\,\overline{\kappa}_{j}}=\left[1_{Q}\right]$
if and only if $\sum_{j=1}^{\rho}\nolimits\overline{\kappa}_{j}\overline{\mathbf{u}}_{j}=0$
since $\mathscr{E}$ is a basis for $Q/{\sim}$ so that $\prod_{\ell=1}^{m}\nolimits\mathsf{E}_{\ell}^{\sum_{j=1}^{\rho}\overline{D}_{\ell j}\overline{\kappa}_{j}}=\left[1_{Q}\right]$
if and only if $\sum_{j=1}^{\rho}\overline{D}_{\ell j}\overline{\kappa}_{j}=0$
for $\ell=1,\ldots,m$. Hence, $\left(\overline{\mathsf{D}}_{1},\ldots,\overline{\mathsf{D}}_{\rho}\right)$
is an independent tuple of dimensions if and only if $\left(\overline{\mathbf{u}}{}_{1},\ldots,\overline{\mathbf{u}}_{\rho}\right)$
is a linearly independent tuple of column vectors. In particular,
the rank of $\left[D_{\ell i}\right]$ is equal to the rank of $\mathscr{D}$,
since the rank of a matrix equals the maximal number of linearly independent
columns.

On the assumption that one exists, we fix an adequate partition $\left\{ \mathscr{A},\mathscr{B},\mathscr{C}\right\} $
of $\mathscr{D}$. For all $\left[q_{i}\right]=\mathsf{D}_{i}$ in
$\mathscr{D}$ there are, by (\ref{eq:exp1-2}), (\ref{eq:trivexp})
and (\ref{eq:dimcol}), $\mathfrak{w}_{i}>0,\mathfrak{w}_{ij}$ such
that
\begin{align*}
\prod_{\ell=1}^{m}\nolimits\mathsf{E}_{\ell}^{D_{\ell i}\mathfrak{w}_{i}} & =\left(\prod_{\ell=1}^{m}\nolimits\mathsf{E}_{\ell}^{D_{\ell i}}\right)^{\!\mathfrak{w}_{i}}=\left[q_{i}\right]^{\mathfrak{w}_{i}}=\prod_{j=1}^{r}\nolimits\left[x_{j}\right]^{\mathfrak{w}_{ij}}=\prod_{j=1}^{r}\nolimits\left(\prod_{\ell=1}^{m}\nolimits\mathsf{E}_{\ell}^{D_{\ell j}}\right)^{\!\mathfrak{w}_{ij}}\\
 & =\prod_{\ell=1}^{m}\nolimits\prod_{j=1}^{r}\nolimits\mathsf{E}_{\ell}^{D_{\ell j}\mathfrak{w}_{ij}}=\prod_{\ell=1}^{m}\nolimits\mathsf{E}_{\ell}^{\sum_{j=1}^{r}D_{\ell j}\mathfrak{w}_{ij}}\qquad\left(i=0,\ldots,n\right).
\end{align*}
As $\left\{ \mathsf{E}_{1},\ldots,\mathsf{E}_{m}\right\} $ is a basis
for $Q/{\sim}$, this means that every system of equations 
\begin{equation}
\begin{cases}
D_{1i}\,\xi_{i}=D{}_{\!11}\,\xi_{i1}+\ldots+D{}_{\!1r}\,\xi_{ir}\\
\cdots & \left(i=0,\ldots,n\right)\\
D_{mi}\,\xi_{i}=D{}_{\!m1}\,\xi_{i1}+\ldots+D{}_{\!mr}\,\xi_{ir}
\end{cases}\label{eq:dimeq-1}
\end{equation}
\textcolor{black}{has a solution $\left(\xi_{i},\xi_{i1},\ldots,\xi_{ir}\right)=\left(\mathfrak{w}_{i},\mathfrak{w}_{i1},\ldots,\mathfrak{w}_{ir}\right)$
such that 
\begin{equation}
\left[q_{i}\right]^{\mathfrak{w}_{i}}=\prod_{j=1}^{r}\nolimits\left[x_{j}\right]^{\mathfrak{w}_{ij}}\qquad\left(\mathfrak{w}_{i}>0\right).\label{eq:qi2}
\end{equation}
Since $\left\{ \mathscr{A},\mathscr{B},\mathscr{C}\right\} $ is a
partition of $\left(q_{1},\ldots,q_{n},q_{0}\right)$, there is some
$i$ such that $\left[y_{k}\right]=\left[q_{i}\right]$ for each $k=0,\ldots,n-r$.
Hence, we can write (\ref{eq:qi2}) as (\ref{eq:exp1-2}), setting
$w_{k}=\mathfrak{w}_{i}$ and $w_{kj}=\mathfrak{w}_{ij}$. For each
$\left[y_{k}\right]$ there is thus, as shown in Section \ref{sec:4},
a unique distinguished tuple of integers $\left(W_{k},W_{k1},\ldots,W_{kr}\right)$
such that $\left[y_{k}\right]^{W_{k}}=\prod_{j=1}^{r}\nolimits\left[x_{j}\right]^{W_{kj}}$,
$W_{k}>0$ and $\gcd\left(W_{k},W_{k1},\ldots,W_{kr}\right)=1$.}

\textcolor{black}{For each adequate partition of $\mathscr{D}$, the
corresponding unique distinguished tuples of integers $\left(W_{k},W_{k1},\ldots,W_{kr}\right)$
can be arranged as rows in a matrix}
\[
\left[\begin{array}{cccc}
W_{0} & W_{01} & \cdots & W_{0r}\\
W_{1} & W_{11} & \cdots & W_{1r}\\
\cdots & \cdots & \cdots & \cdots\\
W_{n-r} & W_{\left(n-r\right)1} & \cdots & W_{\left(n-r\right)r}
\end{array}\right].
\]
Each such matrix gives one equation in (\ref{eq:qeqsyst}) of the
form
\[
y_{0}^{W_{\!\left(\ell\right)0}}=\prod_{j=1}^{r}\nolimits\!x{}_{\left(\ell\right)j}^{W_{\!\left(\ell\right)0j}}\,\Psi_{\ell}\left(\pi{}_{\left(\ell\right)1},\ldots,\pi{}_{\left(\ell\right)n-r}\right),
\]
where $\pi{}_{\left(\ell\right)k}=y{}_{\left(\ell\right)k}^{W_{\!\left(\ell\right)k}}\left(\prod_{j=1}^{r}\nolimits\!x{}_{\left(\ell\right)j}^{W_{\!\left(\ell\right)kj}}\right)^{-1}$
for $k=1,\ldots,n-r$.

In physics\emph{, }$\mathscr{E}=\left\{ \mathsf{E}_{1},\ldots,\mathsf{E}_{m}\right\} $
is the set of dimensions corresponding to a set of base units such
as those in the SI system. From a mathematical point of view, the
basis $\mathscr{E}$ for $Q/{\sim}$ can be chosen freely since the
dependencies described by (\ref{eq:exp1-1-1-1}) are not specified
in terms of a basis. Thus, different dimensional matrices may give
the same set of representations of $\Phi$. It should be noted, though,
that in physics a change of units is often associated with a change
of quantity space \cite{key-26}, affecting dependencies among dimensions
and quantities.

\section{Examples of dimensional analysis\label{sec:6}}

We start with an example of dimensional analysis in a wide sense,
and then proceed to more conventional applications with dimensional
matrices as inputs.
\begin{example}
\label{x1}Consider a dimension tuple $\mathscr{D}=\left(\left[x\right],\left[1_{Q}\right]\right)$,
where $\left[x\right]\neq\left[1_{Q}\right]$. The tuple $\left(\left[x\right]\right)$
is the only maximal independent subtuple of $\mathscr{D}$, and $\left(\left[x\right]\right)$
does not contain $\left[1_{Q}\right]$, so $\mathscr{D}$ has exactly
one adequate partition, 
\[
\left(\left[x\right]\right),\;\emptyset,\;\left(\left[1_{Q}\right]\right).
\]
By Theorem \ref{pimain} any corresponding (necessarily complete)
quantity function of the form $\Phi:\left[x\right]\rightarrow\left[1_{Q}\right]$
that admits a covariant scalar representation is a constant function
since it has a representation of the form
\[
\Phi\left(x\right)=x^{0}\Psi\left(\right)=k,
\]
where $k\in\left[1_{Q}\right]$ is a constant. Had $\Phi$ not admitted
a covariant scalar representation then nothing could have been said
about how $\Phi\left(x\right)$ depends on $x$, so we obtain specific
information about $\Phi$ from this assumption.
\end{example}
\begin{example}
\label{x2}Assume that the time of oscillation $t$ of a pendulum
depends on its length $\ell$, the mass of the bob $m$, the amplitude
of the oscillation $\theta$ (an angle) and the constant of gravity
$g$, that is, $t^{W}=\Phi(\ell,m,\theta,g)$. Let the dependencies
among the dimensions in $\mathscr{D}=\left(\left[\ell\right],\left[m\right],\left[\theta\right],\left[g\right],\left[t\right]\right)$
be given by the dimensional matrix
\[
\begin{array}{cccccc}
 & \left[\ell\right] & \left[m\right] & \left[\theta\right] & \left[g\right] & \left[t\right]\\
\mathsf{L} & 1 & 0 & 0 & 1 & 0\\
\mathsf{T} & 0 & 0 & 0 & -2 & 1\\
\mathsf{M} & 0 & 1 & 0 & 0 & 0
\end{array}.
\]
There is only one maximal independent subtuple of $\mathscr{D}$ not
containing $\left[t\right]$, namely $\left(\left[\ell\right],\left[m\right],\left[g\right]\right)$.
The corresponding adequate partition of $\mathscr{D}$ is 
\[
\left(\left[\ell\right],\left[m\right],\left[g\right]\right),\quad\left(\left[\theta\right]\right),\quad\left(\left[t\right]\right),
\]
and the dimensional model is
\[
\begin{array}{cccc}
 & \left[\ell\right] & \left[m\right] & \left[g\right]\\
\mathsf{L} & 1 & 0 & 1\\
\mathsf{T} & 0 & 0 & -2\\
\mathsf{M} & 0 & 1 & 0
\end{array},\begin{array}{c}
\left[\theta\right]\\
0\\
0\\
0
\end{array},\begin{array}{c}
\left[t\right]\\
0\\
1\\
0
\end{array}.
\]
Clearly, $\left[t\right]^{2}=\left[\ell\right]^{1}\left[m\right]^{0}\left[g\right]^{-1}=\left[\ell\right]^{1}\left[g\right]^{-1}$
and $\left[\theta\right]^{1}=\left[\ell\right]^{0}\left[m\right]^{0}\left[g\right]^{0}=\left[1_{Q}\right]$,
so we have
\begin{equation}
t^{2}=\ell g^{-1}\Psi\left(\theta/1_{Q}\right)=\ell g^{-1}\Psi\left(\theta\right).\label{eq:pendel}
\end{equation}
Thus, the time of oscillation does not really depend on the mass of
the bob. 

Alternatively, choose a local basis $E$ and let $t$, $\ell$, $g$,
$\theta$ denote the real numbers $\mu_{E}\left(t\right),\mu_{E}\left(\ell\right),\mu_{E}\left(g\right),\mu_{E}\left(\theta\right)$,
respectively. Then (\ref{eq:pendel}) becomes $t^{2}=\ell g^{-1}\uppsi\left(\theta\right)$,
where $\uppsi:\mathbb{R}\rightarrow\mathbb{R}$ is a covariant scalar
representation of $\Psi$, and from the assumption that $t,\ell,g>0$
follows that $\uppsi\left(\theta\right)>0$ and that (\ref{eq:pendel})
is equivalent to
\[
t=\sqrt{\ell g^{-1}\uppsi\left(\theta\right)}.
\]
It can be shown \cite{key-25} that $\sqrt{\uppsi\left(\theta\right)}\rightarrow2\pi$
as $\theta\rightarrow0$, so for small oscillations this equation
simplifies to $t=2\pi\sqrt{\ell g^{-1}}$.
\end{example}
\begin{example}[based on a problem in Gibbings \cite{key-9}, pp. 107\textendash 108]
\label{x3}Assume that the force $F$ between two parallel equal-sized
plates of a capacitor depends on the area \emph{a} of each plate,
the distance \emph{z} between the plates, the permittivity $\epsilon$
of the dielectric layer, and the potential difference $\varphi$,
that is, $F^{W}=\Phi\left(a,z,\epsilon,\varphi\right)$, so that $\mathscr{D}=\left(\left[a\right],\left[z\right],\left[\epsilon\right],\left[\varphi\right],\left[F\right]\right)$.
In the $\mathsf{L,T,M,I}$ system of dimensions we have the following
dimensional matrix
\[
\begin{array}{cccccc}
 & \left[a\right] & \left[z\right] & \left[\epsilon\right] & \left[\varphi\right] & \left[F\right]\\
\mathsf{L} & 2 & 1 & -3 & 2 & 1\\
\mathsf{T} & 0 & 0 & 4 & -3 & -2\\
\mathsf{M} & 0 & 0 & -1 & 1 & 1\\
\mathsf{I} & 0 & 0 & 2 & -1 & 0
\end{array}.
\]
Corresponding to the two adequate partitions of $\mathscr{D}$ for
which $\mathscr{C}=\left(\left[F\right]\right)$, 
\[
\left(\left[a\right],\left[\epsilon\right],\left[\varphi\right]\right),\;\left(\left[z\right]\right),\;\left(\left[F\right]\right),\qquad\left(\left[z\right],\left[\epsilon\right],\left[\varphi\right]\right),\;\left(\left[a\right]\right),\;\left(\left[F\right]\right),
\]
there are two dimensional models with $F$ as the dependent variable,
\[
\begin{array}{cccc}
 & \left[a\right] & \left[\epsilon\right] & \left[\varphi\right]\\
\mathsf{L} & 2 & -3 & 2\\
\mathsf{T} & 0 & 4 & -3\\
\mathsf{M} & 0 & -1 & 1\\
\mathsf{I} & 0 & 2 & -1
\end{array},\begin{array}{c}
\left[z\right]\\
1\\
0\\
0\\
0
\end{array},\begin{array}{c}
\left[F\right]\\
1\\
-2\\
1\\
0
\end{array},\qquad\begin{array}{cccc}
 & \left[z\right] & \left[\epsilon\right] & \left[\varphi\right]\\
\mathsf{L} & 1 & -3 & 2\\
\mathsf{T} & 0 & 4 & -3\\
\mathsf{M} & 0 & -1 & 1\\
\mathsf{I} & 0 & 2 & -1
\end{array},\begin{array}{c}
\left[a\right]\\
2\\
0\\
0\\
0
\end{array},\begin{array}{c}
\left[F\right]\\
1\\
-2\\
1\\
0
\end{array},
\]
and these models give the following equations
\[
\begin{cases}
F=\epsilon\varphi^{2}\,\psi_{1}\left(z^{2}/a\right), & \left(E1\right)\\
F=\epsilon\varphi^{2}\,\psi_{2}\left(a/z^{2}\right). & \left(E2\right)
\end{cases}
\]
Note that $\psi_{1}\left(x\right)=\psi_{2}\left(x^{-1}\right)$, so
it suffices to consider either $\psi_{1}$ or $\psi_{2}$. It is known
that $F=\frac{\epsilon\varphi^{2}a}{2z^{2}}$ if edge effects are
disregarded, meaning that $\psi_{2}\left(x\right)\sim\frac{x}{2}$
as $x\rightarrow\infty$.

The rank of the dimensional matrix is 3, and physically there is no
mechanical motion and no flow of current in this example, suggesting
that the dimensions $\mathsf{T}$, $\mathsf{M}$ and $\mathsf{I}$
be replaced by those of force and charge, $\mathsf{F}=\mathsf{L}\mathsf{T}^{-2}\mathsf{M}$
and $\mathsf{Q}=\mathsf{T}\mathsf{I}$. There is indeed a basis $\left\{ \mathsf{L},\mathsf{F},\mathsf{Q},\mathsf{X}\right\} $
for $Q/{\sim}$ such that the dimensions in $\mathscr{D}$ can be
expressed in terms of $\left\{ \mathsf{L},\mathsf{F},\mathsf{Q}\right\} $
(with only zero exponents for $\mathsf{X})$, giving the dimensional
matrix 
\[
\begin{array}{cccccc}
 & \left[a\right] & \left[z\right] & \left[\epsilon\right] & \left[\varphi\right] & \left[F\right]\\
\mathsf{L} & 2 & 1 & -2 & 1 & 0\\
\mathsf{F} & 0 & 0 & -1 & 1 & 1\\
\mathsf{Q} & 0 & 0 & 2 & -1 & 0
\end{array}.
\]
We obtain $\left(E1\right)$ and $\left(E2\right)$ again, as we should,
as a change of basis from $\left\{ \mathsf{L,T,M,I}\right\} $ to
$\left\{ \mathsf{L},\mathsf{F},\mathsf{Q},\mathsf{X}\right\} $ does
not affect the dependencies among the dimensions in $\mathscr{D}$.
\end{example}
\begin{example}
\label{x4}Let $\mathfrak{A}$ and $\mathfrak{B}$ be two bodies of
mass $a$ and $b$, respectively, and let $c$ be the combined mass
of $\mathfrak{A}$ and $\mathfrak{B}$. We are looking for a quantity
function $\Phi$ such that $c^{W}=\Phi(a,b)$, so $\mathscr{D}=\left(\left[a\right],\left[b\right],\left[c\right]\right)$.
The simple dimensional matrix is
\[
\begin{array}{cccc}
 & \left[a\right] & \left[b\right] & \left[c\right]\\
\mathsf{M} & 1 & 1 & 1
\end{array},
\]
and, corresponding to the two adequate partitions of $\mathscr{D}$
for which $\mathscr{C}=\left(\left[c\right]\right)$,
\[
\left(\left[a\right]\right),\;\left(\left[b\right]\right),\;\left(\left[c\right]\right),\qquad\left(\left[b\right]\right),\;\left(\left[a\right]\right),\;\left(\left[c\right]\right),
\]
there are two dimensional models with $c$ as dependent variable,
\[
\begin{array}{cc}
 & \left[a\right]\\
\mathsf{M} & 1
\end{array},\begin{array}{c}
\left[b\right]\\
1
\end{array},\begin{array}{c}
\left[c\right]\\
1
\end{array},\qquad\begin{array}{cc}
 & \left[b\right]\\
\mathsf{M} & 1
\end{array},\begin{array}{c}
\left[a\right]\\
1
\end{array},\begin{array}{c}
\left[c\right]\\
1
\end{array},
\]
 giving the system of equations
\[
\begin{cases}
c=a\,\Psi_{1}\left(b/a\right), & (M1)\\
c=b\,\Psi_{2}\left(a/b\right). & (M2)
\end{cases}
\]
Thus, $a\Psi_{1}\left(b/a\right)=b\Psi_{2}\left(a/b\right)$, and
if we assume by symmetry that $\Phi(a,b)=\Phi(b,a)$ then $a\Psi_{1}\left(b/a\right)=\Phi(a,b)=\Phi(b,a)=a\Psi_{2}\left(b/a\right)$,
so $\Psi_{1}=\Psi_{2}$. Setting $x=b/a$ we thus obtain a functional
equation of the form $\Psi\left(x\right)=F\left(x,\Psi\right)$, namely
\[
\Psi\left(x\right)=x\,\Psi\left(x^{-1}\right),
\]
which has solutions of the form 
\[
\Psi(x)=k\left(1_{Q}+x\right)\qquad\left(k,x\in\left[1_{Q}\right]\right).
\]

If $a\neq0_{\mathsf{M}}$ then $c^{W}=\Phi(a,b)$ can be represented
as $\left(M1\right)$, so if $\Phi(a,0_{\mathsf{M}})=a$ then $a=a\Psi(0_{\mathsf{\left[1_{Q}\right]}})=ak1_{Q}$,
so $k=1_{Q}$. Substituting $\Psi$ for $\Psi_{1}$ in $\left(M1\right)$,
we thus obtain
\[
c=k\left(a+b\right)=a+b,
\]
as one might expect. By symmetry, $\Phi(0_{\mathsf{M}},b)=b$ for
$b\neq0_{\mathsf{M}}$, so if $\Phi(0_{\mathsf{M}},0_{\mathsf{M}})=0_{\mathsf{M}}$
then $c=a+b$ for all $a,b\in\mathsf{M}$ (but negative masses may
not exist physically). 

We note that while a quantity function $\Phi$ is represented as a
product of quantities in dimensional analysis, such a product can
sometimes be rewritten as a sum of quantities. It is also clear that
addition of masses exemplifies a general principle: if $a,b,\Phi\left(a,b\right)\in\mathsf{X\neq\left[1_{Q}\right]}$
and $\Phi\left(a,b\right)=\Phi\left(b,a\right)$ then $\Phi\left(a,b\right)=k\left(a+b\right)$.
\end{example}
\begin{example}[based on a problem in Buckingham \cite{key-2}, pp. 358\textendash 359]
\label{x5}

It is assumed that the energy density $u$ at a fixed point is determined
by the strengths $E$ and $H$ of an electric field $\boldsymbol{\mathbf{E}}$
and a magnetic $\mathbf{H}$-field, respectively, as well as the permittivity
$\epsilon$ and permeability $\mu$ of the medium, that is, $u^{W}=\Phi\left(E,H,\epsilon,\mu\right)$.
The dimensional matrix corresponding to $\mathscr{D}=\left(\left[E\right],\left[H\right],\left[\epsilon\right],\left[\mu\right],\left[u\right]\right)$
in the $\mathsf{L},\mathsf{T},\mathsf{M},\mathsf{I}$ system is
\[
\begin{array}{cccccc}
 & \left[E\right] & \left[H\right] & \left[\epsilon\right] & \left[\mu\right] & \left[u\right]\\
\mathsf{L} & 1 & -1 & -3 & 1 & -1\\
\mathsf{T} & -3 & 0 & 4 & -2 & -2\\
\mathsf{M} & 1 & 0 & -1 & 1 & 1\\
\mathsf{I} & -1 & 1 & 2 & -2 & 0
\end{array}.
\]
This matrix has rank 3, and four adequate partitions of $\left(\left[E\right],\left[H\right],\left[\epsilon\right],\left[\mu\right],\left[u\right]\right)$
with $\mathscr{C}=\left(\left[u\right]\right)$ exist, namely 
\begin{gather*}
\left(\left[E\right],\left[\epsilon\right],\left[\mu\right]\right),\;\left(\left[H\right]\right),\;\left(\left[u\right]\right),\qquad\left(\left[H\right],\left[\epsilon\right],\left[\mu\right]\right),\;\left(\left[E\right]\right),\;\left(\left[u\right]\right),\\
\left(\left[E\right],\left[H\right],\left[\epsilon\right]\right),\;\left(\left[\mu\right]\right),\;\left(\left[u\right]\right),\qquad\left(\left[E\right],\left[H\right],\left[\mu\right]\right),\;\left(\left[\epsilon\right]\right),\;\left(\left[u\right]\right).
\end{gather*}
The corresponding equations are
\[
\begin{cases}
u=E^{2}\epsilon\,\Psi_{1}\left(H^{2}/\left(E^{2}\epsilon\mu^{-1}\right)\right), & \left(EM1\right)\\
u=H^{2}\mu\,\Psi_{2}\left(E^{2}/\left(H^{2}\epsilon^{-1}\mu\right)\right), & \left(EM2\right)\\
u=E^{2}\epsilon\,\Psi_{1}'\left(\mu/\left(E^{2}H^{-2}\epsilon\right)\right), & \left(EM3\right)\\
u=H^{2}\mu\,\Psi_{2}'\left(\epsilon/\left(E^{-2}H^{2}\mu\right)\right). & \left(EM4\right)
\end{cases}
\]
We note that only two distinct functions, say $\Psi_{1}$ and $\Psi_{2}$,
occur in the representations of $\Phi$. Also, $E,H,\epsilon,\mu$
occur only in the combinations $E'=\epsilon E^{2}$ and $H'=\mu H^{2}$
in these representations, so we can write $u^{W}=\Phi\left(E,H,\epsilon,\mu\right)$
as $u^{W}=\Phi'\left(E',H'\right)$, with corresponding dimension
tuple $\mathscr{D}'=\left(\left[E'\right],\left[H'\right],\left[u\right]\right)$
and dimensional matrix
\[
\begin{array}{cccc}
 & \left[E'\right] & \left[H'\right] & \left[u\right]\\
\mathsf{L} & -1 & -1 & -1\\
\mathsf{T} & -2 & -2 & -2\\
\mathsf{M} & 1 & 1 & 1\\
\mathsf{I} & 0 & 0 & 0
\end{array}.
\]
Thus, $E',H',u\in\mathsf{L}^{-1}\mathsf{T}^{-2}\mathsf{M}$ and there
are two adequate partitions of $\mathscr{D}'$ with $u$ as dependent
variable, 
\[
\left(\left[E'\right]\right),\;\left(\left[H'\right]\right),\;\left(\left[u\right]\right),\qquad\left(\left[H'\right]\right),\;\left(\left[E'\right]\right),\;\left(\left[u\right]\right),
\]
 and two corresponding representations of $\Phi'$,
\[
\begin{cases}
u=E'\,\Psi_{1}'\left(H'/E'\right), & (EM1')\\
u=H'\,\Psi_{2}'\left(E'/H'\right). & (EM2')
\end{cases}
\]
Buckingham also finds these representations, writing them as $u=\epsilon E^{2}\,\varphi_{1}\left(\frac{\mu H^{2}}{\epsilon E^{2}}\right)$
and $u=\mu H^{2}\,\varphi_{2}\left(\frac{\epsilon E^{2}}{\mu H^{2}}\right)$
\cite[p. 359]{key-2}. He then remarks:
\begin{quotation}
\noindent {\small{}Assuming }that{\small{} the complete formula is
\[
u=\frac{1}{8\pi}\left(\epsilon E^{2}+\mu H^{2}\right)
\]
 we have
\[
\varphi_{1}\left(x\right)=\varphi_{2}\left(x\right)=\frac{1+x}{8\pi}.
\]
}{\small \par}
\end{quotation}
In dimensional analysis, we are not supposed to know ''the complete
formula'' at the outset, however, so let us reverse this inference.
We have $E'\,\Psi_{1}'\left(H'/E'\right)=H'\,\Psi_{2}'\left(E'/H'\right)$,
and if we assume, invoking a physical symmetry between $\mathbf{E}$
and $\mathbf{H}$, that $\Phi'\left(E',H'\right)=\Phi'\left(H',E'\right)$
so that $E'\,\Psi_{1}'\left(H'/E'\right)=E'\,\Psi_{2}'\left(H'/E'\right)$
then we conclude that $\Psi_{1}'=\Psi_{2}'$, so setting $x=H'/E'$
we obtain the functional equation $\Psi\left(x\right)=x\,\Psi\left(x^{-1}\right)$
as in Example \ref{x4}. Again, $\Psi\left(x\right)=k\left(1_{Q}+x\right)$,
where $k,x\in\left[1_{Q}\right]$, and substituting this in $\left(EM1'\right)$
or $\left(EM2'\right)$ we obtain
\[
u=k\left(\epsilon E^{2}+\mu H^{2}\right),
\]
or $u=k\left(\epsilon E^{2}+\mu^{-1}B^{2}\right)$ if $B=\mu H$.

Working with examples, Buckingham thus recognised in \cite{key-2}
that there may be more than one way of representing $\Phi$. However,
he dismissed this observation by asserting that then the representations
are ''equivalent'' \cite[p. 359, 362]{key-2}, implying that it
suffices to consider any one of them. 
\end{example}
\begin{example}[based on a problem in Bridgman \cite{key-10}, pp.~5\textendash 8]
\label{x6} Let two bodies $\mathfrak{B}$ and $\mathfrak{b}$ with
masses $M$ and $m$ revolve around each other under influence of
their mutual gravitational attraction, as in the classical two-body
problem. Let $t$ denote the time of revolution and $d$ the mean
distance between $\mathfrak{B}$ and $\mathfrak{b}$ (or another characteristic
distance). One might want to find out how $t$ depends on $M$, $m$
and $d$, \textcolor{black}{but there is no adequate partition of
$\left(\left[M\right],\left[m\right],\left[d\right],\left[t\right]\right)$
such that $\mathscr{C}=\left(\left[t\right]\right)$ since $\left(\mathsf{L},\mathsf{T},\mathsf{M}\right)$
is an independent dimension tuple. There is thus no local basis $E$
such that $\mu_{E}\left(t^{W}\right)$ is defined, so $\Phi_{0}:\left[M\right]\times\left[m\right]\times\left[d\right]\rightarrow\left[t^{W}\right]$
does not have a ''physically meaningful'' \textendash{} that is,
covariant \textendash{} scalar representation. Bridgman suggests that
$t$ does also depend on the gravitational constant }\textcolor{black}{\emph{G}}\textcolor{black}{,
that is, $t^{W}=\Phi\left(M,m,d,G\right)$, giving the dimensional
matrix}
\[
\begin{array}{cccccc}
 & \left[M\right] & \left[m\right] & \left[d\right] & \left[G\right] & \left[t\right]\\
\mathsf{L} & 0 & 0 & 1 & 3 & 0\\
\mathsf{T} & 0 & 0 & 0 & -2 & 1\\
\mathsf{M} & 1 & 1 & 0 & -1 & 0
\end{array}.
\]
 There are two adequate partitions of $\left(\left[t\right],\left[M\right],\left[m\right],\left[d\right],\left[G\right]\right)$
with $\mathscr{C}=\left(\left[t\right]\right)$, 
\[
\left(\left[M\right],\left[d\right],\left[G\right]\right),\;\left(\left[m\right]\right),\;\left(\left[t\right]\right),\qquad\left(\left[m\right],\left[d\right],\left[G\right]\right),\;\left(\left[M\right]\right),\;\left(\left[t\right]\right),
\]
 so we obtain a system of two equations,
\[
\begin{cases}
t^{2}=M^{-1}d^{3}G^{-1}\,\Psi_{1}\left(m/M\right), & (K1)\\
t^{2}=m^{-1}d^{3}G^{-1}\,\Psi_{2}\left(M/m\right). & (K2)
\end{cases}
\]
Hence, $M^{-1}\Psi_{1}\left(m/M\right)=m^{-1}\Psi_{2}\left(M/m\right)$,
and for symmetry reasons we may assume that $\Phi\left(M,m,d,G\right)=\Phi\left(m,M,d,G\right)$
so that $M^{-1}\Psi_{1}\left(m/M\right)=M^{-1}\Psi_{2}\left(m/M\right)$.
This implies $\Psi_{1}=\Psi_{2}$, so setting $x=m/M$ we obtain the
functional equation
\[
\Psi\left(x\right)=x^{-1}\,\Psi\left(x^{-1}\right).
\]
This functional equation has solutions of the form 
\[
\Psi(x)=k\left(1_{Q}+x\right)^{-1}\qquad\left(k,x\in\left[1_{Q}\right]\right),
\]
and substituting this in either ($K1$) or ($K2$) gives
\begin{equation}
t^{2}=kd^{3}G^{-1}(M+m)^{-1}.\label{eq:kepler}
\end{equation}

Here, $k$ and $G$ are constants, so if $M+m\approx M$ and $M$
is constant (several planets orbit the sun) then, approximately, $t^{2}\propto d^{3}$;
this is Kepler's third law. 

As before, (\ref{eq:kepler}) can be interpreted both as a quantity
equation and as a scalar equation, and the latter can also be written
as 
\begin{equation}
t=c\sqrt{d^{3}G^{-1}(M+m)^{-1}}.\label{eq:scalOrb}
\end{equation}

It is worth pointing out that Bridgman considered only one equation,
namely 
\[
t=\frac{r^{\frac{3}{2}}}{G^{\frac{1}{2}}m_{2}^{\frac{1}{2}}}\phi\left(\frac{m_{2}}{m_{1}}\right),
\]
corresponding to ($K1$) with $\phi\left(m_{2}/m_{1}\right)=\Psi_{1}\left(m_{1}/m_{2}\right)$
or ($K2$) with $\phi\left(m_{2}/m_{1}\right)=\Psi_{2}\left(m_{1}/m_{2}\right)$.
The basic reason why Bridgman was not able to derive the much more
informative equation (\ref{eq:scalOrb}) was that he did not reflect
on the possibility that the original function could have more than
one representation, and as a consequence he did not reflect on what
inferences could be drawn from symmetries between different representations.
Remarkably, the same restricted way of thinking still dominates dimensional
analysis, next to a century after the appearance of Bridgman's\emph{
}classic \cite{key-10}, but mathematics does not always move quickly.
\end{example}

\section{Ex nihilo nihil fit\label{sec:7}}

Some examples of dimensional analysis give the impression that it
can be used to derive physical laws ''out of nothing''. This is
of course an illusion; while the assumptions may be hidden or so intuitive
as to be overlooked, assumptions there are. Recall, in particular,
that the application of Theorem \ref{pimain} to dimensional analysis
is based on the premise that the quantity function $\Phi$ posited
has a covariant scalar representation. Thus, to assume that $\zeta=\upphi\left(\zeta_{1},\ldots,\zeta_{n}\right)$
is a ''physically meaningful'' scalar equation, so that we can obtain
an equation $\zeta=\prod_{j=1}^{r}\xi_{j}^{W_{0j}}\uppsi\left(\pi{}_{1},\ldots,\pi{}_{n-r}\right)$
by dimensional analysis, is to assume that $\upphi$ is a covariant
scalar representation of a quantity function $\Phi.$ This assumption,
underlying dimensional analysis, is a general covariance principle
about the equivalence of certain reference frames, defined by corresponding
systems of units of measurement.

Also, stronger assumptions will in general lead to stronger results.
For example, to derive (\ref{eq:kepler}) in Example \ref{x6} we
had to assume that $\Phi\left(M,m,d,G\right)=\Phi\left(m,M,d,G\right)$.
The symmetry between $\mathfrak{B}$ and $\mathfrak{b}$ motivating
this assumption would in turn seem to require that a description of
the situation in which $\mathfrak{B}$ is at rest and $\mathfrak{b}$
revolves around $\mathfrak{B}$ is equivalent to a description in
which $\mathfrak{b}$ is at rest and $\mathfrak{B}$ revolves around
$\mathfrak{b}$, so that we may reverse the roles of $\mathfrak{B}$
and $\mathfrak{b}$. This is in fact a deep symmetry assumption, related
to Mach's principle. 

\section{Conclusion}

In modern algebra, an abstract vector is the real thing while its
coordinates is the shadow of the vector, dependent not only on the
vector itself but also on the location of the source of illumination
\textendash{} a particular choice of basis. The notion of quantity
calculus \cite{key-2-1} derives from an analogous duality between
the quantities measured and the numbers serving as their measures,
and an idea of the primacy of quantities. From this perspective, we
need a new $\pi$ theorem which is about quantities, not primarily
about numbers. Theorems \ref{pimain} and \ref{piAlt} are reformulations
of the $\pi$ theorem designed to meet this need.

The proposed reconceptualisation requires a mathematical definition
of the notion of a quantity. As defined here, a quantity is simply
an element of a quantity space, defined, in turn, in \cite{key-7}
and Section \ref{sec:2}. The operations and identities characterising
quantity spaces define a quantity calculus suitable for formulating
a quantity counterpart of the $\pi$ theorem. 

As we have seen, the fact that a theorem about scalars (measures),
the $\pi$ theorem, is replaced by a theorem about quantities as a
basis for dimensional analysis does not mean that dimensional analysis
can no longer yield scalar equations, but the new way of thinking
leads to other consequences and benefits that concern the logic, results
and practice of dimensional analysis.

Despite its long history and widespread use, dimensional analysis
remains somewhat unsettled and controversial. There is apparently
still room for foundational work, and the rigorous algebraic foundation
provided by the theory of quantity spaces and Theorem \ref{pimain}
would seem to clarify essential mathematical aspects of dimensional
analysis. In particular, a precise mathematical meaning is given to
the notion of a ''physically meaningful'' scalar equation.

Secondly, the new approach to dimensional analysis redefines the results
that can be expected from it. Instead of representing $\phi$ such
that $\phi\left(t_{1},\ldots,t_{n}\right)=0$ by $\psi$ such that
$\psi\left(\pi_{1},\ldots,\pi_{n-r}\right)=0$, as in Vaschy's and
Buckingham's original $\pi$ theorem, we prefer to represent $\phi$
such that $t^{W}=\phi\left(t_{1},\ldots,t_{n}\right)$ by means of
$\psi$ such that $t^{W}=\prod_{j=1}^{r}\nolimits\!x{}_{j}^{W_{j}}\,\psi\left(\pi_{1},\ldots,\pi_{n-r}\right)$.
Even more significantly, we derive an equation system of representations
rather than a single representation. Such an equation system can then
make it possible to extract additional information about $\phi$ that
could not have been obtained by means of any single equation, or even
all equations considered individually.

Finally, it should be noted that we have presented an algorithmic
technique for dimensional analysis. We have described rules that make
it possible to generate unique equation systems from a dimensional
matrix, obtaining one system for each mathematically legitimate choice
of dependent variable. It would not be difficult to create a computer
implementation of this algorithm.

\end{document}